\documentclass[11pt,letterpaper]{article}

\usepackage[utf8]{inputenc}
\usepackage[english]{babel}

\usepackage{subcaption}

\usepackage[margin=1in]{geometry}
\usepackage{amsfonts,amsmath,amsthm,thmtools}
\usepackage{nicefrac}

\usepackage[dvipsnames,usenames]{xcolor}
\usepackage[colorlinks=true,pdfpagemode=UseNone,urlcolor=RoyalBlue,linkcolor=RoyalBlue,citecolor=OliveGreen,pdfstartview=FitH]{hyperref}
\usepackage[capitalize,noabbrev]{cleveref}

\declaretheorem{theorem}
\declaretheorem[sibling=theorem]{lemma}
\declaretheorem[sibling=theorem]{corollary}
\declaretheorem[sibling=theorem]{fact}
\declaretheorem[sibling=theorem]{claim}

\usepackage{todonotes}

\usepackage{booktabs}
\usepackage{caption}
\usepackage[numbers, sort]{natbib}

\usepackage{tcolorbox}

\usepackage{enumitem}

\usepackage{algorithm}
\usepackage{algpseudocode}

\newcommand{\E}{\mathbf{E}}
\renewcommand{\Pr}{\mathbf{Pr}}

\newcommand{\opt}{\textsc{Opt}}
\newcommand{\alg}{\textsc{Alg}}

\newcommand{\Sref}[1]{\hyperref[#1]{\S\ref*{#1}}}

\newcommand{\p}[1]{\left( #1 \right)}
\newcommand{\cur}[1]{\left\{ #1 \right\}}

\newcommand{\fb}{\bar{f}}
\newcommand{\gb}{\bar{g}}

\newcommand{\ddt}{\frac{\mathrm{d}}{\mathrm{d}t}}
\newcommand{\eps}{\varepsilon}

\title{Edge-weighted Online Stochastic Matching Under Jaillet-Lu LP}

\author{
    Shuyi Yan
    \thanks{BARC, University of Copenhagen. Supported by VILLUM Foundation Grant 54451. Email: \texttt{shya@di.ku.dk}.}
}
\date{}

\begin{document}
    
\begin{titlepage}
    \thispagestyle{empty}
    \maketitle
    \begin{abstract}
        \thispagestyle{empty}
        The online stochastic matching problem was introduced by \cite{feldman2009online}, together with the $(1-\frac1e)$-competitive Suggested Matching algorithm. In the most general edge-weighted setting, this ratio has not been improved for more than one decade, until recently \cite{yan2024edge} beat the $1-\frac1e$ bound and \cite{qiu2023improved} further improved it to $0.650$. Both works measure the online competitiveness against the offline LP relaxation introduced by Jaillet and Lu \cite{jaillet2014online}. The same LP has also played an important role in other settings as it is a natural choice for two-choice online algorithms.

In this paper, we prove an upper bound of $0.663$ and a lower bound of $0.662$ for edge-weighted online stochastic matching under Jaillet-Lu LP.
We propose a simple hard instance and identify the optimal online algorithm for this specific instance which has a competitive ratio of $<0.663$.
Despite the simplicity of the instance, we then show that a near-optimal algorithm for it, which has a competitive ratio of $>0.662$, can be generalized to work on all instances without any loss.

As our algorithm is generalized from a real near-optimal algorithm instead of manually combining trivial strategies, it has two natural advantages compared with previous works: (1) its matching strategy varies from time to time; (2) it utilizes global information about offline vertices.
On the other hand, the upper bound suggests that more powerful LPs and multiple-choice strategies are needed if we want to further improve the ratio by $>0.001$.

In addition to our main result, we also generalize the asymptotic equivalence between the Poisson arrival model and the original online stochastic matching established by \cite{huang2021online}, removing the requirement of approximate monotonicity for the online algorithm.
    \end{abstract}
\end{titlepage}

\section{Introduction}

The online bipartite matching problem was introduced by \cite*{karp1990optimal}, motivated by applications such as online advertising. In this problem, we are given a bipartite graph. Advertisers are modeled as \emph{offline vertices} on one side of the graph, which are known at the beginning. Impressions (user queries) are modeled as \emph{online vertices} on the other side, which arrive one by one. The weights on the edges indicate how much the advertisers are interested in the impressions. When an online vertex arrives, the online algorithm sees its incident edges, and must immediately and irrevocably decide how to match it. \cite{karp1990optimal} proposed the celebrated Ranking algorithm, which achieves a competitive ratio of $1-\frac1e \approx 0.632$ in the unweighted setting. This is optimal in the \emph{worst case model}, where we assume the online algorithm has no prior information about the graph.

However, in real-world scenarios, the algorithm often knows some prior information. For example, online advertising platforms can use historical data to predict user behaviors in the future. \cite*{feldman2009online} introduced the \emph{online stochastic matching} model, where we assume all online vertices are drawn independently from an identical distribution which is known by the algorithm in the beginning. \cite{feldman2009online} proposed the Suggested Matching algorithm which is $(1-\frac1e)$-competitive even in the most difficult edge-weighted setting, and then beat the $1-\frac1e$ bound in the unweighted setting under the integral assumption\footnote{It means that the expected number of online vertices of each type is an integer, or equivalently, is $1$. It makes the problem easier since the lower bound of $1$ implies that online vertices of one type cannot be almost fully matched by only one offline vertex.}. Further works removed the integral assumption \cite{manshadi2012online, jaillet2014online, huang2021online, huang2022power, tang2022fractional} and beat the $1-\frac1e$ bound in the vertex-weighted setting \cite{huang2021online, huang2022power, tang2022fractional} and the free-disposal\footnote{It allows us to match an offline vertex many times and only count the most valuable one in the end.} edge-weighted setting \cite{huang2022power}.

A common framework of the online algorithms is to solve an LP to get a fractional matching which upper-bounds the offline optimal solution, and then use the fractional solution to guide the online strategy. The online competitiveness is then measured against the LP relaxation. In particular, the LP proposed by Jaillet and Lu \cite*{jaillet2014online} is quite useful for two-choices algorithms (which most previous algorithms are) in the sense that it limits the case that an online vertex type is mostly matched to one offline vertex in the fractional solution, forcing most online vertices to have at least two choices. \emph{Under Jaillet-Lu LP}, \cite*{huang2022power} provided both a lower bound and an upper bound of the same ratio $0.706$\footnote{The exact value is $1 - \frac{1}{1-\ln 2} \big( \frac{1}{2e} - \frac{\ln 2}{e^2} \big)$.} for unweighted, vertex-weighted and free-disposal edge-weighted settings.

The edge-weighted setting was known to be \emph{strictly} harder than the above settings due to an upper bound of $0.703$ \cite{huang2022power}. Recently, \cite*{yan2024edge} beat the $1-\frac1e$ bound with a competitive ratio of $0.645$ by extending the Suggested Matching algorithm to have different strategies in $3$ different time periods. \cite*{qiu2023improved} improved the ratio to $0.650$ by using $5$ periods. Both algorithms still use the trivial strategy from Suggested Matching in more than half of the time. In contrast, instead of using several fixed strategies in different periods, the state-of-the-art Poisson Online Correlated Selection algorithm \cite{huang2022power} for unweighted/vertex-weighted settings has a strategy that changes smoothly over time. A natural question arises: Can smoothly changed strategies help to further improve the competitive ratio in the edge-weighted setting?

On the other hand, both algorithms \cite{yan2024edge,qiu2023improved} use the Jaillet-Lu LP relaxation. We know that the provably best ratio for online algorithms under Jaillet-Lu LP is $0.706$ in the other $3$ settings. What about the edge-weighted setting?

\subsection{Our Results}

In this paper, we prove that for edge-weighted online stochastic matching, the optimal competitive ratio against Jaillet-Lu LP is between $0.662$ and $0.663$.

We first propose a hard instance and show that, for this instance, the optimal online algorithm can only achieve a competitive ratio $<0.663$. Then, we extract some key properties of a near-optimal matching process of this specific instance which are sufficient to guarantee a competitive ratio $>0.662$. Finally, we propose a generalized algorithm which maintains these properties for all instances, which is therefore $0.662$-competitive for all instances. The algorithm improves the best known ratio of $0.650$ \cite{qiu2023improved}, and the hardness result suggests that more powerful LPs are necessary for further improvements.\footnote{In the unweighted and vertex-weighted settings, \cite{huang2022power} goes beyond the optimal algorithm against Jaillet-Lu LP by utilizing more powerful LPs. In the free-disposal edge-weighted setting, there is still no algorithm achieving this.}

Compared with existing algorithms \cite{yan2024edge,qiu2023improved} for edge-weighted online stochastic matching, our algorithm has two major differences: (1) instead of only using one or several fixed strategies, we deploy a strategy that smoothly varies from time to time; (2) when an online vertex arrives, instead of only looking at its neighbors locally, our algorithm utilizes some global information about the probability distribution of remaining offline vertices.
Intuitively speaking, both are natural movements towards the optimal algorithm.

Below, we briefly summarize our high-level ideas. For readers not familiar with this problem, detailed definitions and explanations can be found in \cref{sec:preliminaries}.

\paragraph{Hardness.}

Previous work \cite{yan2024edge} showed that, under Jaillet-Lu LP, the problem can be reduced to the case that there are only two classes of online vertex types:
\begin{itemize}
    \item A first-class type has only one offline neighbor;
    \item A second-class type has only two offline neighbors, and it is equally matched to them.
\end{itemize}
In the hard instance, we want to have as many (arrival rates of) first-class types as possible, since they don't have a second choice when the only neighbor has been matched. For second-class types, we want to maximize the probability that, given that one neighbor has been matched, the other neighbor has also been matched. So, if two second-class types share one common neighbor, we tend to let them share another. The above ideas give us the following instance:
\begin{itemize}
    \item For each offline vertex, there is a first-class online vertex type matched to it, with an arrival rate as high as possible;
    \item Offline vertices are divided into pairs. For each pair, there is a second-class online vertex type matched to them.
\end{itemize}
In \cref{sec:hardness}, we identify the optimal online algorithm for this instance.
By carefully setting the edge weights, we show that the optimal algorithm achieves a competitive ratio $<0.663$.

\paragraph{Algorithm.}

After getting the optimal algorithm for the specific hard instance, we try to design our algorithm by generalizing the special algorithm to work on all instances.
For each first-class (resp. second-class) edge, we want to match it at the same rate as the first-class (resp. second-class) edges in the hard instance.
First-class vertices are easy to deal with since they have only one possible choice when arriving, so we focus on second-class vertices.
When a second-class vertex arrives, to make sure that it is feasible to match it at the desired rate, we (at least) need an upper bound of the probability that the two offline neighbors have both been matched.

To be able to simulate the algorithm on all instances, this probability must take its maximum value in our hard instance, so that the upper bound works for all instances. In other words, roughly speaking, the algorithm should match two offline vertices in the hard instance in the most positively correlated way.
Unfortunately, the optimal algorithm for the hard instance does not have this property. Therefore, in \cref{sec:algorithm}, we introduce a near-optimal algorithm with this property which has a competitive ratio $>0.662$ on the hard instance.
Then we show that the near-optimal algorithm can be generalized to have the same competitive ratio on all instances.

\subsection{Other Related Work}

In the worst case model, the Ranking algorithm \cite{karp1990optimal} and its generalization \cite{aggarwal2011online} achieved the optimal competitive ratio of $1-1/e$ for unweighted and vertex-weighted matching, respectively. For free-disposal edge-weighted matching, a sequence of works \cite{fahrbach2020edge,gao2022improved,blanc2022multiway} improved the competitive ratio to $0.5239$.

In the random order model\footnote{The adversary decides the graph, then online vertices arrive in a uniformly random order.}, the currently best ratios are $0.696$ \cite{goel2008online,mahdian2011online} for unweighted matching and $0.686$ \cite{huang2019online,jin2021improved,peng2025revisiting} for vertex-weighted matching.
For edge-weighted matching, \cite{kesselheim2013optimal} achieved the optimal competitive ratio of $1/e$.

In the online stochastic matching model, for unweighted and vertex-weighted matching, a sequence of works \cite{feldman2009online,bahmani2010improved,manshadi2012online,jaillet2014online,huang2021online,huang2022power,tang2022fractional} improved the competitive ratio to $0.716$ and the upper bound to $0.823$. For free-disposal edge-weighted matching, \cite{huang2022power} proposed a $0.706$-competitive algorithm.
Under the integral assumption, the state-of-the-art algorithms achieved competitive ratios of $0.7299$ for unweighted matching \cite{manshadi2012online,jaillet2014online,brubach2020online}, $0.725$ for vertex-weighted matching \cite{jaillet2014online}, and $0.705$ for edge-weighted matching \cite{haeupler2011online,brubach2020online}.
An upper bound of $0.703$ \cite{huang2022power} for (general) edge-weighted matching inferred that it is strictly harder than the above settings.
For non i.i.d. arrivals, recent works \cite{tang2022fractional,chen2024stochastic} achieved a competitive ratio of $0.69$ for unweighted and vertex-weighted matching.

\section{Preliminaries}
\label{sec:preliminaries}

\paragraph{Edge-weighted Online Stochastic Matching.}

An instance of edge-weighted online stochastic matching consists of:
\begin{itemize}
    \item A bipartite graph $G=(I,J,E)$, where $I$ is the set of online vertex types and $J$ is the set of offline vertices;
    \item An arrival rate $\lambda_i$ for each online vertex type $i\in I$;
    \item A nonnegative weight $w_{ij}$ for each edge $(i,j)\in E$.
\end{itemize}
$\Lambda=\sum_{i\in I}\lambda_i$ is the total number of online vertices. $\Lambda$ online vertices will arrive one by one. When an online vertex arrives, it independently draws a random type $i\in I$ with probability $\lambda_i/\Lambda$. Then, the online algorithm can either match it to an unmatched offline neighbor, or simply discard it. The decision has to be made immediately and irrevocably. The objective is to maximize the expected total weight of all matched edges in the end.
The \emph{competitive ratio} of the online algorithm is the infimum of the ratio of its expected objective to the expected objective of the offline optimal matching. We also say that an algorithm is $\alpha$-competitive if its competitive ratio is at least $\alpha$.

\subparagraph{Poisson Arrival Model.}

In the Poisson arrival model, instead of assuming a deterministic number of online vertices, we assume that online vertices arrive following a Poisson process with arrival rate $\Lambda$ in the time horizon $[0,1]$. An equivalent view is that each online vertex type $i\in I$ independently follows a Poisson process with arrival rate $\lambda_i$.
This equivalence and the continuous arrival process greatly help the analyses of online stochastic matching algorithms, since \cite{huang2021online} established an asymptotic equivalence between the Poisson arrival model and the original model for large $\Lambda$.

However, the asymptotic equivalence in \cite{huang2021online} actually does not hold for arbitrary algorithms, as it requires some additional property of the online algorithm. In particular, it works for all existing unweighted and vertex-weighted algorithms but not for existing edge-weighted algorithms.
Here we state a more general asymptotic equivalence between the two models, which holds for arbitrary edge-weighted (and therefore vertex-weighted and unweighted) algorithms. See \cref{app:equivalence} for the proof and more discussions.

\begin{theorem}
\label{thm:equiv}
    For any $\alpha$-competitive online algorithm in one of the Poisson arrival model and the original model of edge-weighted online stochastic matching, there is an online algorithm (with the same time complexity) in the other model that is $(\alpha-o(1))$-competitive for $\Lambda\to\infty$.
\end{theorem}

In this paper, we present our upper and lower bounds in the Poisson arrival model. We remark that prior works in the original model also conventionally assume sufficiently large $\Lambda$, or explicitly have the $o(1)$ additive term.

\paragraph{Jaillet-Lu LP.}

The (Poisson arrival version of) Jaillet-Lu LP \cite{jaillet2014online} is as follows:
\begin{alignat}{2}
    \mbox{maximize}\quad & \sum_{(i,j)\in E} w_{ij} x_{ij} & {} & \nonumber \\
    \mbox{subject to}\quad & \sum_{j \in J} x_{ij}\leq\lambda_i &\quad& \forall i\in I \nonumber \\
    & \sum_{i\in I} x_{ij} \leq 1 &\quad& \forall j\in J \nonumber \\
    & \sum_{i\in I} \max\{2x_{ij}-\lambda_i,0\} \leq 1-\ln2 &\quad& \forall j\in J \label{eqn:yj} \\
    & x_{ij} \geq 0 &\quad& \forall(i,j)\in E \nonumber
\end{alignat}

A solution of the Jaillet-Lu LP is a fractional matching between online vertex types and offline vertices. \cite{huang2021online} proved that the optimal value of the LP, that is, the objective of the optimal fractional matching, upper-bounds the expected objective of the optimal offline integral matching.
In fact, each $x_{ij}$ upper-bounds the probability that offline vertex $j$ is matched to an online vertex of type $i$.
The reason behind \cref{eqn:yj} is that, an online vertex type cannot put too much of its arrival rate to the same offline vertex since the probability that it arrives (at least once) is smaller than the arrival rate.

In this paper, we consider the competitive ratio against the optimal value of the Jaillet-Lu LP instead of the optimal integral matching.
Without loss of generality, we define $x_{ij}=0$ for any pair $(i,j)\notin E$. For each online vertex type $i$, define $x_i=\sum_{j\in J}x_{ij}$. For each offline vertex $j$, define $x_j=\sum_{i\in I}x_{ij}$.

\paragraph{Instance Reduction.} \cite{yan2024edge} showed that, any instance $A$ can be efficiently transformed into another instance $B$ such that the optimal fractional matching (under Jaillet-Lu LP) of $B$ satisfies the following conditions:
\begin{itemize}
    \item Each offline vertex is fully matched. In other words, $x_j=1$ for each offline vertex $j$.
    \item There are only two classes of online vertex types:
    \begin{itemize}
        \item A \emph{first-class} online vertex type $i$ only has one offline neighbor $j$, where $x_{ij}=x_i=\lambda_i$.
        \item A \emph{second-class} online vertex type $i$ has exactly two offline neighbors $j_1$ and $j_2$, where $x_{ij_1}=x_{ij_2}=\lambda_i/2$ and $x_i=\lambda_i$.
    \end{itemize}
\end{itemize}
In addition, for any online algorithm that works for $B$, using the transformation as a preprocessing step, it also works for $A$ with the same competitive ratio.

The transformation works roughly as follows. We first add dummy vertices so that $x_i=\lambda_i$ for all $i\in I$ and $x_j=1$ for all $j\in J$. Then we ``split'' each online vertex type into first-class and second-class types in a way that minimizes the arrival rates of first-class types, which will preserve \cref{eqn:yj}.
For more details, we refer the reader to \cite{yan2024edge} Section 3.

Without loss of generality, in this paper we restrict the instance to satisfy the above conditions.
If an edge is incident to a first-class (resp. second-class) online vertex type, we call it a \emph{first-class} (resp. \emph{second-class}) edge.
For each offline vertex $j$, let $y_j$ denote the fraction that $j$ is matched to first-class vertices, i.e.
$$y_j=\sum_{i:\text{first-class}}x_{ij}.$$
Then \cref{eqn:yj} means that $y_j\le 1-\ln 2$.

\section{Hardness}
\label{sec:hardness}

We define our hard instance $G$ with $\Lambda=2$ as follows. There are two offline vertices and a second-class online vertex type, with arrival rate $2\ln 2$, connected to both offline vertices with unit edge weights. In addition, each offline vertex is connected to a first-class online vertex type with arrival rate $1-\ln2$, with some edge weight $k\ge 1$. We will choose $k$ later. See \cref{fig:hard} for an illustration.
We note that we can also create large hard instances by simply duplicating $G$ many times, i.e., our hardness result will hold for graphs with arbitrarily large $\Lambda$.

\begin{figure}[!h]
    \centering
    \includegraphics[scale=0.5]{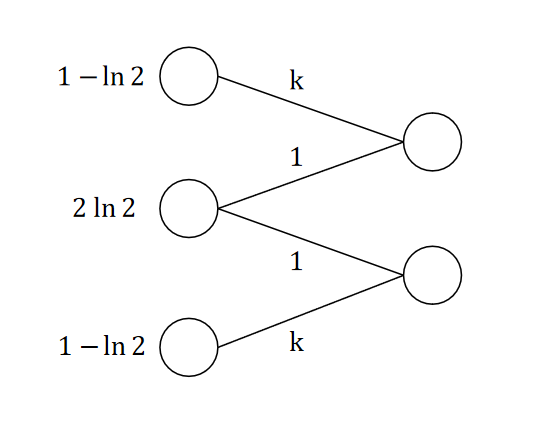}
    \caption{The hard instance $G$. The number next to each online vertex type is the arrival rate. The number on each edge is the edge weight.}
    \label{fig:hard}
\end{figure}

It's easy to check that in the optimal Jaillet-Lu LP solution, all vertices in $G$ will be fully matched, which gives us the following fact.

\begin{fact}
\label{fact:LP-value}
    The optimal value of Jaillet-Lu LP on $G$ is $2\ln2 + (2-2\ln2)k$.
\end{fact}

Now we are going to find the optimal online algorithm on $G$. Clearly, since $k\ge 1$, when a first-class online vertex arrives, we shall always match it if we can. What about a second-class online vertex? Intuitively, the decision should depend on $k$ and the current time $t$. When $t$ is small but $k$ is large, we may want to discard this online vertex even if it has unmatched neighbors, to preserve the chance of matching a first-class vertex in the future. Also notice that the number of unmatched neighbors may also affect our decision, since it affects the chance that the offline vertex can be matched with another second-class online vertex in the future.

Let's first look at an easier case where we don't need to consider the number of unmatched neighbors. Let $G'$ be the remaining part of $G$ after one offline vertex is matched. See \cref{fig:hard-sub}.

\begin{figure}[!h]
    \centering
    \includegraphics[scale=0.5]{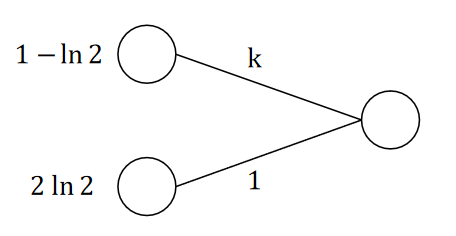}
    \caption{The sub-instance $G'$. The online vertex type with no unmatched neighbor can be ignored.}
    \label{fig:hard-sub}
\end{figure}

\begin{algorithm}[!h]
\caption{Optimal Algorithm on $G'$}
\label{alg:opt-G'}
    \begin{algorithmic}[1]
        \State When a first-class online vertex arrives, match it to the neighbor if the neighbor is unmatched.
        \State When a second-class online vertex arrives at time $t$, match it to the neighbor if the neighbor is unmatched and $t>t_1$.
    \end{algorithmic}
\end{algorithm}

\begin{lemma}
\label{lem:opt-G'}
    For any $k\ge 1$, there exists $0\le t_1\le 1$ such that \cref{alg:opt-G'} is the optimal algorithm on $G'$.
\end{lemma}

\begin{proof}
    Let $\opt(G',t)$ denote the expected objective of the optimal online algorithm running on $G'$ in the time interval $[t,1]$. Clearly, $\opt(G',t)$ is at most $k$ and decreasing on $t$. Consider what the optimal algorithm should do when an online vertex arrives at time $t$ and the offline vertex is still unmatched. If the algorithm decides to discard the vertex, it will gain $\opt(G',t)$ in the future by definition. Since $\opt(G',t)\le k$, the algorithm can always decide to match the vertex if it is first-class. If it is second-class, the algorithm can match it if $\opt(G',t)\le 1$. Since $\opt(G',t)$ is decreasing, $t_1$ can be set as the boundary where $\opt(G',t)$ starts going below $1$.
\end{proof}

Clearly, the optimal algorithm on $G$ will degenerate to \cref{alg:opt-G'} after one offline vertex is matched. When both offline vertices are unmatched, it needs another parameter $t_0$.

\begin{algorithm}[!h]
\caption{Optimal Algorithm on $G$}
\label{alg:opt-G}
    \begin{algorithmic}[1]
        \State When a first-class online vertex arrives, match it to the neighbor if the neighbor is unmatched.
        \State When a second-class online vertex with exactly one unmatched neighbor arrives at time $t$, match it to the neighbor if $t>t_1$.
        \State When a second-class online vertex with two unmatched neighbors arrives at time $t$, match it to an arbitrary neighbor if $t>t_0$.
    \end{algorithmic}
\end{algorithm}

\begin{lemma}
\label{lem:opt-G}
    For any $k\ge 1$, there exists $0\le t_0\le t_1\le 1$ such that \cref{alg:opt-G} is the optimal algorithm on $G$.
\end{lemma}

\begin{proof}
    When one of the offline vertices has been matched, $G$ will degenerate to $G'$, and \cref{alg:opt-G} will also degenerate to \cref{alg:opt-G'}. \cref{lem:opt-G'} shows that \cref{alg:opt-G'} is optimal in this case for some $t_1$.

    Then we analyze the strategy of the optimal online algorithm when both offline vertices are unmatched. Let $\opt(G,t)$ denote the expected objective of the optimal online algorithm running on $G$ in the time interval $[t,1]$. Notice that $\opt(G,t)\le 2\opt(G',t)$, since $2\opt(G',t)$ can be viewed as the optimal expected objective in the case that a second-class online vertex is allowed to be simultaneously matched to both neighbors on $G$.
    
    When a first-class online vertex arrives, if the algorithm decides to match it, the expected objective in the end will be $\opt(G',t)+k \ge 2\opt(G',t) \ge \opt(G,t)$, so the algorithm can always match it.
    
    When a second-class online vertex arrives, the algorithm can match it if $\opt(G',t)+1 \ge \opt(G,t)$, and it doesn't matter which neighbor is matched by symmetry. When $t>t_1$, we know $\opt(G',t) \le 1$, then $\opt(G',t)+1 \ge 2\opt(G',t) \ge \opt(G,t)$. So the algorithm can always match it when $t>t_1$. Then, it remains to show that $\opt(G,t)-\opt(G',t)$ is decreasing in $[0,t_1]$. If this is true, then we can set $t_0$ as the boundary where $\opt(G,t)-\opt(G',t)$ starts going below $1$.

    Denote the two offline vertices by $u$ and $v$. The basic idea is that, by the symmetry between $u$ and $v$, we can introduce some constraints to the optimal algorithm to force $u$ to contribute exactly $\opt(G',t)$ to the expected objective. Then, $v$ must contribute $\opt(G,t)-\opt(G',t)$. We can show that for $t'<t$, there is an algorithm for $t'$ which simulates the optimal algorithm for $t$ on $v$ while still having optimal performance on $u$.

    Note that, by \cref{lem:opt-G'}, when $t<t_1$, the optimal algorithm on $G'$ will not match any second-class vertex. So, we can force the optimal algorithm on $G$ to never match $u$ before time $t_1$. On the other hand, since the algorithm will decide to match a second-class online vertex after $t_1$, we can also force it to match to $u$ as long as possible. With these two constraints, the algorithm still achieves an expected objective of $\opt(G,t)$, and on $u$ it actually acts the same as the optimal algorithm on $G'$, which means $u$ contributes exactly $\opt(G',t)$ to the expected objective.

    Now consider an algorithm for some $t'<t$ which satisfies the two constraints and always match a first-class vertex if possible. It's straight-forward to see that $u$ must contribute $\opt(G',t')$. On the other hand, at any time in $[t,1]$, the probability that $u$ has been matched in our algorithm is at least the probability that $u$ has been matched in the optimal algorithm for $t$. So, we can simply do not match $v$ in $[t',t)$, and then in $[t,1]$ match $v$ in the same way as the optimal algorithm for $t$. Then the expected objective of our algorithm is $\opt(G',t')+\opt(G,t)-\opt(G',t)$, which means $\opt(G,t')-\opt(G',t') \ge \opt(G,t)-\opt(G',t)$. In other words, $\opt(G,t)-\opt(G',t)$ is decreasing in $[0,t_1]$ as desired.
\end{proof}

Combining \cref{fact:LP-value,lem:opt-G}, we get our hardness result.

\begin{theorem}
\label{thm:hardness}
    For $k\approx 3.40216$, the optimal online algorithm has a competitive ratio $\approx 0.66275 < 0.663$ against the optimal value of the Jaillet-Lu LP on $G$.
\end{theorem}

\begin{proof}
    For any $k,t_0,t_1$, let $\alg(k,t_0,t_1)$ denote the expected objective of \cref{alg:opt-G}. According to our matching strategy, $\alg(k,t_0,t_1)$ can be calculated using straightforward calculus. However, the closed-form expression is extremely long, so we omit it here. See \cref{app:hardness} for computational details.

    By \cref{fact:LP-value}, the optimal LP value is $2\ln2+(2-2\ln2)k$. So the competitive ratio is $\frac{\alg(k,t_0,t_1)}{2\ln2+(2-2\ln2)k}$.    
    For $k\approx 3.40216$, the ratio achieves its maximum value $\approx 0.66275$ at $t_0\approx 0.12437$ and $t_1\approx 0.29539$. 
\end{proof}

\section{Algorithm}
\label{sec:algorithm}

In this section, we are going to generalize a special algorithm on our hard instance $G$ to get the same competitive ratio on all possible instances. The basic idea is, for each first-class (resp. second-class) edge, we want to match it at the same rate\footnote{The ``match rate'' means the derivative of the probability that the vertex/edge has been matched at time $t$.} as we match a first-class (resp. second-class) edge of $G$. Note that this also fixes the match rate of each offline vertex.

Is it possible? Consider any pair of offline vertices. Without loss of generality, we can assume there is a second-class online vertex type (with an arrival rate of almost $0$) adjacent to them. Intuitively, to make sure that we can match them at the desired rate, we wish that these two offline vertices are matched in a very negatively correlated way, so that there is enough probability that the second-class online vertex has at least one unmatched neighbor. In other words, a necessary condition is, $G$ must be the instance that the pair of offline vertices are matched in the most positively correlated way (with respect to the match rate of each first-class/second-class edge, which is decided by the algorithm).

It's not hard to see that in \cref{alg:opt-G}, unfortunately, the two offline vertices are actually matched in the most negatively correlated way during time $[t_0,t_1]$: When one of them is matched, the other must not be matched by a second-class online vertex. However, during other times, they are matched in the most positively correlated way, as we will see later.\footnote{During time $[0,t_0]$, they are matched independently, which is also the most positively correlated way given the match rate of each first-class/second-class edge, since the match rate of each second-class edge is $0$.} So, if we restrict $t_0=t_1$ to remove the middle phase, we will have the desired property.

\begin{algorithm}[!h]
\caption{Restricted version of \cref{alg:opt-G}}
\label{alg:opt-G-res}
    \begin{algorithmic}[1]
        \State When a first-class online vertex arrives, match it to the neighbor if the neighbor is unmatched.
        \State When a second-class online vertex with at least one unmatched neighbor arrives at time $t>t_0$, match it to a uniformly random unmatched neighbor.
    \end{algorithmic}
\end{algorithm}

If we pick the optimal $t_0$, \cref{alg:opt-G-res} is only slightly worse than \cref{alg:opt-G}.

\begin{lemma}
\label{lem:res}
    For $t_0\approx 0.14753$, \cref{alg:opt-G-res} matches each first-class edge (in $G$) with probability $\Gamma(1-\ln2)$ and each second-class edge with probability $\Gamma\ln2$, where $\Gamma \approx 0.66217 > 0.662$.
\end{lemma}

\begin{proof}
    Given that we have analyzed the entire matching process of \cref{alg:opt-G} in \cref{app:hardness}, it is straightforward to verify the lemma. See \cref{app:alg-res} for computational details.
\end{proof}

There is another issue for generalizing \cref{alg:opt-G-res}. In the hard instance $G$, every offline vertex $j$ satisfies $y_j=1-\ln2$, i.e. it is maximally matched to first-class vertices. However, in general we may have various $y_j\in[0,1-\ln2]$. As we have different strategies for first-class and second-class vertices, offline vertices with different $y_j$ will have different match rates. More importantly, the match rates will be different from the offline vertices in $G$.

There is no reason that $y_j=1-\ln2$ is not the worst case. The idea to solve this issue is that, we can manually force some second-class edges to be redefined as first-class edges. While these edges are still incident to second-class vertices, we can make them act as first-class edges by not allowing them to steal the other half of the arrival rate of the incident online vertex when the other offline neighbor has been matched, as we will see later. So, for each offline vertex $j$ such that $y_j<1-\ln2$, we pick some incident second-class edges with a total matching fraction of $1-\ln2-y_j$\footnote{Of course, we may have to pick ``a part of'' an edge to guarantee that the total fraction is exactly that value. This is not a problem, since an online vertex type of arrival rate $\lambda$ is equivalent to two online types of arrival rates $\lambda_1$ and $\lambda_2$ with the same offline neighbors, where $\lambda_1+\lambda_2=\lambda$. So we only need to split one online type into two for each offline vertex.}, and redefine them as first-class edges.

Now, a first-class edge can be incident to a second-class vertex, but we successfully guarantee that each offline vertex is matched to first-class edges with an exact $1-\ln2$ fraction.

The last thing before seeing the generalized algorithm is to define some important notations. Let $g(t)$ denote the probability that, when running \cref{alg:opt-G-res} on $G$, both offline vertices have been matched at time $t$.
By symmetry, the two offline vertices are matched at the same rate. For any of them, let $f(t)$ denote the probability that it has been matched at time $t$.
Let $\gb(t)=1-g(t)$, which is the probability that at least one offline vertex is unmatched at time $t$. Similarly, let $\fb(t)=1-f(t)$.

When applying our generalized algorithm on any instance, we will use $f_u(t)$ to denote the probability that the offline vertex $u$ has been matched at time $t$, and $g_{u,v}(t)$ to denote the probability that both $u$ and $v$ have been matched at time $t$. Let $\fb_u(t)=1-f_u(t)$ and $\gb_{u,v}(t)=1-g_{u,v}(t)$.

\begin{algorithm}[!h]
\caption{Generalized version of \cref{alg:opt-G-res}}
\label{alg:general}
    \begin{algorithmic}[1]
        \State When a first-class online vertex arrives, match it to the only neighbor (if possible).
        \State When a second-class online vertex arrives, match it along each first-class incident edge (if possible) with probability $\frac12$.
        \State When a second-class online vertex arrives at time $t>t_0$, match it along each second-class incident edge (if possible) with probability $\frac{\gb(t)}{2\gb_{u,v}(t)}$ if the other offline neighbor is unmatched, or with probability $\frac{\gb(t)}{\gb_{u,v}(t)}$ if the other offline neighbor is matched.
    \end{algorithmic}
\end{algorithm}

Before proving that \cref{alg:general} simulates \cref{alg:opt-G-res} well, we need to make sure that for each second-class online vertex, the total probability we assigned to its incident edges does not exceed $1$. This can be guaranteed by the following claim.

\begin{claim}
\label{lem:valid}
    When running \cref{alg:general}, the probability $\frac{\gb(t)}{\gb_{u,v}(t)}$ is always at most $1$. 
\end{claim}

To prove \cref{lem:valid}, we need an auxiliary algorithm. Let $g'_{u,v}(t)$ denote the probability that both $u$ and $v$ are unmatched at time $t$. Let $g'(t)$ denote the probability that, when running \cref{alg:opt-G-res} on $G$, both offline vertices are unmatched at time $t$. Note that $\gb(t)=2\fb(t)-g'(t)$.

\begin{algorithm}[!h]
\caption{Auxiliary version of \cref{alg:general}}
\label{alg:aux}
    \begin{algorithmic}[1]
        \State When a first-class online vertex arrives, match it to the only neighbor (if possible).
        \State When a second-class online vertex arrives, match it along each first-class incident edge (if possible) with probability $\frac12$.
        \State When a second-class online vertex arrives at time $t>t_0$, match it along each second-class incident edge (if possible) with probability $\frac12\min\{\frac{\gb(t)}{2\fb_u(t)-g'_{u,v}(t)},1\}$ if the other offline neighbor is unmatched, or with probability $\min\{\frac{\gb(t)}{2\fb_u(t)-g'_{u,v}(t)},1\}$ if the other offline neighbor is matched.
    \end{algorithmic}
\end{algorithm}

It's clear that \cref{alg:aux} is a valid algorithm. We will show that, when running \cref{alg:aux}, we always have $2\fb_u(t)-g'_{u,v}(t) \ge \gb(t)$ and $2\fb_u(t)-g'_{u,v}(t) = \gb_{u,v}(t)$. This means \cref{alg:aux} is equivalent to \cref{alg:general}, and proves \cref{lem:valid}.

\begin{lemma}
\label{lem:fu-le}
    When running \cref{alg:aux}, we always have $f_u(t)\le f(t)$.
\end{lemma}

\begin{proof}
    Clearly, $f_u(t)=f(t)$ when $t\le t_0$. So we start with $f_u(t_0)=f(t_0)$ and only consider $t>t_0$.
    For each online neighbor $i$ of $u$, let $p_{i,u}(t)$ denote the probability that, when an online vertex of type $i$ arrives at time $t$, it is matched with $u$. By definition,
    $$\ddt f_u(t)=\sum_{i}\lambda_i p_{i,u}(t).$$
    
    By the definition of \cref{alg:general}, it's easy to see that $p_{i,u}(t)=\fb_u(t)$ for each first-class neighbor $i$ and $p_{i,u}(t)=\frac12\fb_u(t)$ for each second-class neighbor $i$ such that $(i,u)$ is a first-class edge. So, for each first-class edge $(i,u)$, we have $\lambda_ip_{i,u}(t)=x_{iu}\fb_u(t)$.

    On the other hand, consider any second-class edge $(i,u)$. Let the other neighbor of $i$ be $v$. Since the probability that $u$ is unmatched but $v$ is matched at time $t$ is $\fb_u(t)-g'_{u,v}(t)$, we have
    \begin{align*}
        p_{i,u}(t)
        & = g'_{u,v}(t)\cdot\frac12\min\cur{\frac{\gb(t)}{2\fb_u(t)-g'_{u,v}(t)},1} + (\fb_u(t)-g'_{u,v}(t))\min\cur{\frac{\gb(t)}{2\fb_u(t)-g'_{u,v}(t)},1} \\
        & = \frac12(2\fb_u(t)-g'_{u,v}(t))\min\cur{\frac{\gb(t)}{2\fb_u(t)-g'_{u,v}(t)},1} \\
        & \le \frac12\gb(t),
    \end{align*}
    so $\lambda_ip_{i,u}(t) \le x_{iu}\gb(t)$.
    Then we have
    $$\ddt f_u(t) \le (1-\ln2)\fb_u(t)+(\ln2)\gb(t).$$
    
    Similarly, by the definition of \cref{alg:opt-G-res}, we have
    \begin{align*}
        \ddt f(t)
        & = (1-\ln2)\fb(t) + (2\ln2)\p{g'(t)\cdot\frac12+(\fb(t)-g'(t))\cdot1} \\
        & = (1-\ln2)\fb(t) + (\ln2)(2\fb(t)-g'(t)) \\
        & = (1-\ln2)\fb(t) + (\ln2)\gb(t).
    \end{align*}
    So:
    \begin{align*}
        \ddt \p{f(t)-f_u(t)}
        & \ge (1-\ln2)(\fb(t)-\fb_u(t)) \\
        & = -(1-\ln2)(f(t)-f_u(t)).
    \end{align*}
    Recall that $f(t_0)-f_u(t_0)=0$, so $f(t)-f_u(t)\ge0$ for $t\ge t_0$.
\end{proof}

The above proof also gives us the following corollary.

\begin{corollary}
\label{cor:fu}
    When running \cref{alg:aux}, if we have $2\fb_u(t)-g'_{u,v}(t) \ge \gb(t)$ for any $t\le t'$, then $f_u(t') = f(t')$.
\end{corollary}

\begin{proof}
    Just notice that, the inequalities in the proof of \cref{lem:fu-le} become equalities when $2\fb_u(t)-g'_{u,v}(t) \ge \gb(t)$.
\end{proof}

\begin{lemma}
\label{lem:guv}
    When running \cref{alg:aux}, we always have $g_{u,v}(t)\le g(t)$.
\end{lemma}

\begin{proof}
    Clearly, $g_{u,v}(t)=g(t)$ when $t\le t_0$. So we start with $g_{u,v}(t_0)=g(t_0)$ and only consider $t>t_0$. When running \cref{alg:opt-G-res}, the probability that exactly one offline vertex is unmatched at time $t$ is $2f(t)-g(t)$.
    In this case, the other one will be matched by both neighbors. So
    $$\ddt g(t)=(2f(t)-g(t))\cdot(1+\ln2).$$

    Similarly, when running \cref{alg:aux}, the probability that exactly one of $u$ and $v$ is unmatched at time $t$ is $f_u(t)+f_v(t)-g_{uv}(t)$. By \cref{lem:fu-le}, this is at most $2f(t)-g_{uv}(t)$. Also notice that the total arrival rate of the neighbors of any offline vertex is $1+\ln2$. So
    $$\ddt g_{uv}(t) \le (2f(t)-g_{uv}(t))\cdot(1+\ln2).$$
    So we have:
    $$\ddt \p{g(t)-g_{uv}(t)} \ge -(1+\ln 2)(g(t)-g_{uv}(t)).$$
    Recall that $g(t_0)-g_{uv}(t_0)=0$, so $g(t)-g_{uv}(t) \ge 0$ for $t\ge t_0$.
\end{proof}

\begin{lemma}
\label{lem:prob}
    When running \cref{alg:aux}, we always have $2\fb_u(t)-g'_{u,v}(t) \ge \gb(t)$.
\end{lemma}

\begin{proof}
    Clearly, $2\fb_u(t)-g'_{u,v}(t)=\gb(t)$ when $t\le t_0$. So we start with $2\fb_u(t_0)-g'_{u,v}(t_0)=\gb(t)$ and only consider $t>t_0$.
    By definition, we have $\gb_{u,v}(t)=\fb_u(t)+\fb_v(t)-g'_{u,v}(t)$. Then, by \cref{lem:guv}, we have $\fb_u(t)+\fb_v(t)-g'_{u,v}(t) \ge \gb(t)$. It means that, for any $t$, either $2\fb_u(t)-g'_{u,v}(t) \ge \gb(t)$, or $2\fb_v(t)-g'_{u,v}(t) \ge \gb(t)$, or we have both.
    
    Notice that the above functions and their derivatives are all continuous in $[t_0,1]$. Let $t'\in[t_0,1]$ be the latest time such that for any $t\in[0,t']$ we have both $2\fb_u(t)-g'_{u,v}(t) \ge \gb(t)$ and $2\fb_v(t)-g'_{u,v}(t) \ge \gb(t)$. It suffices to show a contradiction for $t'<1$. If $t'<1$, let $\eps$ be a sufficiently small constant such that, without loss of generality, $2\fb_u(t)-g'_{u,v}(t) < \gb(t)$ and $2\fb_v(t)-g'_{u,v}(t) \ge \gb(t)$ for $t\in(t',t'+\eps]$.
    Then, since $2\fb_v(t)-g'_{u,v}(t) \ge \gb(t)$ for $t\le t'+\eps$, we have $f_v(t)=f(t)$ for $t\le t'+\eps$ by \cref{cor:fu}. Combining with \cref{lem:fu-le}, we get $\fb_u(t)\ge\fb_v(t)$ for $t\le t'+\eps$. Then we have:
    $$2\fb_u(t)-g'_{u,v}(t) \ge \fb_u(t)+\fb_v(t)-g'_{u,v}(t) \ge \gb(t)$$
    for $t\le t'+\eps$, which is a contradiction.
\end{proof}

Combining \cref{lem:prob,cor:fu} directly gives us the following corollary.

\begin{corollary}
\label{cor:fu-2}
    When running \cref{alg:aux}, we always have $f_u(t) = f(t)$.
\end{corollary}

This then gives us the following corollary.

\begin{corollary}
\label{cor:gg}
    When running \cref{alg:aux}, we always have $2\fb_u(t)-g'_{u,v}(t) = \gb_{u,v}(t)$.
\end{corollary}

\begin{proof}
    By definition, $\gb_{u,v}(t)=\fb_u(t)+\fb_v(t)-g'_{u,v}(t)$. By \cref{cor:fu-2}, $\fb_u(t)=\fb_v(t)=f(t)$.
\end{proof}

As we said before, combining \cref{lem:prob,cor:gg} proves \cref{lem:valid} and also shows that \cref{alg:aux} is equivalent to \cref{alg:general}. Therefore, the above conclusions (in particular, \cref{cor:fu}) also work for \cref{alg:general}. We summarize this in the following lemma.

\begin{lemma}
\label{lem:fu}
    \cref{alg:general} is a valid algorithm. Furthermore, when running \cref{alg:general}, we always have $f_u(t) = f(t)$.
\end{lemma}

We are actually quite close to the end. It only remains to calculate the match rate of each edge from \cref{lem:fu}. For each edge $(i,u)$, let $q_{i,u}(t)$ denote the probability that $(i,u)$ has been matched at time $t$. Let $q_1(t)$ (resp. $q_2(t)$) denote the probability that, when running \cref{alg:opt-G-res} on $G$, a particular first-class (resp. second-class) edge has been matched at time $t$.

\begin{lemma}
\label{lem:edge}
    When running \cref{alg:general}, for any first-class edge $(i,u)$, we always have $\frac{q_{i,u}(t)}{x_{iu}}=\frac{q_1(t)}{1-\ln2}$; for any second-class edge $(i,u)$, we always have $\frac{q_{i,u}(t)}{x_{iu}}=\frac{q_2(t)}{\ln2}$.
\end{lemma}

\begin{proof}
    Consider any first-class edge $(i,u)$ such that $i$ is first-class, which means $\lambda_i=x_{iu}$. At time $t$, the probability that $u$ is unmatched is $\fb_u(t)$, which equals $\fb(t)$ by \cref{lem:fu}. When $u$ is unmatched and $i$ arrives, we will match $(i,u)$ with probability $1$. So we have
    $$\ddt q_{i,u}(t) = \fb(t)\cdot\lambda_i\cdot1 = x_{iu}\fb(t).$$
    Similarly, 
    $$\ddt q_1(t) = (1-\ln2)\fb(t).$$
    In addition, they all start with $q_{i,u}(0)=q_1(0)=0$. Therefore $\frac{q_{i,u}(t)}{x_{iu}}=\frac{q_1(t)}{1-\ln2}$.
    
    Consider any first-class edge $(i,u)$ such that $i$ is second-class, which means $\lambda_i=2x_{iu}$. When $u$ is unmatched and $i$ arrives, we will match $(i,u)$ with probability $1/2$. So we have
    $$\ddt q_{i,u}(t) = \fb(t)\cdot\lambda_i\cdot\frac12 = x_{iu}\fb(t).$$
    So we still have $\frac{q_{i,u}(t)}{x_{iu}}=\frac{q_1(t)}{1-\ln2}$.

    Now, consider any second-class edge $(i,u)$, where $i$ must be second-class. Let the other neighbor of $i$ be $v$.
    At time $t$, the probability that both $u$ and $v$ are unmatched is $g'_{u,v}(t)$. In this case, when $i$ arrives, we will match $(i,u)$ with probability $\frac{\gb(t)}{2\gb_{u,v}(t)}=\frac{\gb(t)}{2(2\fb_u(t)-g'_{u,v}(t))}$.
    The probability that $u$ is unmatched but $v$ has been matched is $\fb_u(t)-g'_{u,v}(t)$. In this case, when $i$ arrives, we will match $(i,u)$ with probability $\frac{\gb(t)}{2\fb_u(t)-g'_{u,v}(t)}$. So we have
    \begin{align*}
        \ddt q_{i,u}(t)
        & = g'_{u,v}(t)\cdot\lambda_i\cdot\frac{\gb(t)}{2(2\fb_u(t)-g'_{u,v}(t))} + (\fb_u(t)-g'_{u,v}(t))\cdot\lambda_i\cdot\frac{\gb(t)}{2\fb_u(t)-g'_{u,v}(t)} \\
        & = x_{i,u}\cdot\frac{\gb(t)}{2\fb_u(t)-g'_{u,v}(t)}\cdot(g'_{u,v}(t)+2(\fb_u(t)-g'_{u,v}(t))) \\
        & = x_{i,u}\gb(t).
    \end{align*}
    Similarly,
    \begin{align*}
        \ddt q_2(t)
        & = g'(t)\cdot(2\ln2)\cdot\frac12 + (\fb(t)-g'(t))\cdot(2\ln2)\cdot1 \\
        & = (\ln2)(g'(t)+2(\fb(t)-g'(t))) \\
        & = (\ln2)\gb(t).
    \end{align*}
    In addition, they all start with $q_{i,u}(0)=q_2(0)=0$. Therefore $\frac{q_{i,u}(t)}{x_{iu}}=\frac{q_2(t)}{\ln2}$.
\end{proof}

\begin{theorem}
    \cref{alg:general} is $\Gamma$-competitive against the optimal value of the Jaillet-Lu LP, where $\Gamma \approx 0.66217 > 0.662$.
\end{theorem}

\begin{proof}
    \cref{lem:res} tells us that $\frac{q_1(1)}{1-\ln2}=\frac{q_2(1)}{\ln2}=\Gamma \approx 0.66217$. Combining it with \cref{lem:edge}, we know that every edge $(i,j)$ is matched with probability $\Gamma x_{ij}$. Then the expected objective of \cref{alg:general} is $\Gamma\sum_{(i,j)\in E}w_{ij}x_{ij}$. Comparing it with the optimal LP value $\sum_{(i,j)\in E}w_{ij}x_{ij}$ proves the theorem.
\end{proof}

Finally, we remark that, to implement \cref{alg:general}, we need to compute $\gb(t)$ and $\gb_{u,v}(t)$ efficiently. In \cref{app:hardness}, we have the closed-form expression of $g(t)$ which can be computed in $O(1)$ time. However, $\gb_{u,v}(t)$ is related to the instance itself, and it might take exponential time for the algorithm to compute it precisely. Nevertheless, we can use the standard Monte-Carlo method to estimate it with an $o(1)$ multiplicative error in polynomial time. It then gives us a polynomial-time $(\Gamma-o(1))$-competitive algorithm. See \cref{app:estimate} for more details.

\section*{Acknowledgments}

We thank Mikkel Thorup and the anonymous reviewers for helpful comments.

\bibliographystyle{alpha}
\bibliography{ref}

\appendix

\section{Appendix}

\subsection{Asymptotic Equivalence Between the Models}
\label{app:equivalence}

Fix an instance of edge-weighted online stochastic matching.
Let $\opt_O$ and $\opt_P$ denote the expected objectives of the optimal offline algorithm in the original model and the Poisson arrival model, respectively.
\cite{huang2021online} proved the following lemma. We note that, although \cite{huang2021online} only considered vertex-weighted instances, their proof (of this lemma) also works for edge-weighted instances.

\begin{lemma}[\cite{huang2021online} Theorem 1]
\label{lem:equiv-opt}
    $(1-o(1))\opt_O \le \opt_P \le \opt_O$.
\end{lemma}

However, for objectives of an online algorithm, their similar result only holds when the algorithm is $\beta$-approximately monotone\footnote{An algorithm is $\beta$-approximately monotone if the expected weight from matching the $k$-th online vertex is at most $\beta$ times that from matching the $l$-th online vertex for any $k>l$.} for $\beta=o(\sqrt\Lambda)$.
Although their argument also works for edge-weighted instances, none of the existing edge-weighted algorithms are finitely approximately monotone. So, we cannot directly apply their asymptotic equivalence to the algorithms in \cite{yan2024edge,qiu2023improved} and this paper.

Nevertheless, the following two lemmas show that approximate monotonicity is not a necessary requirement for the online algorithm. Combining them with \cref{lem:equiv-opt}, we get \cref{thm:equiv}.

\begin{lemma}
    For any online algorithm $A_O$ in the original model, there is an online algorithm $A_P$ (with the same time complexity) in the Poisson arrival model, such that for any instance, if $A_O$ has an expected objective of $\alg_O$, then $A_P$ has an expected objective of at least $(1-o(1))\alg_O$.
\end{lemma}

\begin{proof}
    Let $n\sim\mathrm{Poisson}(\Lambda)$ denote the number of online vertices arriving in the Poisson arrival model (which is unknown to $A_P$ in the beginning). Let $m=(1-\Lambda^{-\frac13})\Lambda$.
    
    Given $A_O$, we define $A_P$ as follows. We uniformly sample $m$ different integers $k_1<k_2<\cdots<k_m$ from $[1,\Lambda]$ in the beginning, then simulate $A_O$ and match the $i$-th online vertex in the same way as the $k_i$-th online vertex in $A_O$ for $i\le m$, and discard remaining online vertices (if there are).
    Specifically, when the $i$-th online vertex arrives ($i\le m$), we simulate $A_O$ as follows. (Let $k_0=0$.) For every $j\in(k_{i-1},k_i)$, we manually sample a dummy online vertex according to the i.i.d. distribution, and let $A_O$ treat it as the $j$-th online vertex. Then, let $A_O$ treat the real arriving online vertex as the $k_i$-the online vertex, and we match or discard it in the same way.

    It is straightforward to see that, for any given $n\ge m$, the expected objective of $A_P$ is $\frac{m}{\Lambda}\alg_O=(1-\Lambda^{-\frac13})\alg_O$. On the other hand, since $\E[n]=\textrm{Var}[n]=\Lambda$, we have $\Pr[n<m] \le \Lambda^{-\frac13}$.
    So the objective of $A_P$ is at least $(1-O(\Lambda^{-\frac13}))\alg_O$.
\end{proof}

\begin{lemma}
    For any online algorithm $A_P$ in the Poisson arrival model, there is an online algorithm $A_O$ (with the same time complexity) in the original model, such that for any instance, if $A_P$ has an expected objective of $\alg_P$, then $A_O$ has an expected objective of at least $(1-o(1))\alg_P-o(1)\opt_O$.
\end{lemma}

\begin{proof}
    Let $n\sim\mathrm{Poisson}(\Lambda)$ denote the number of online vertices arriving in the Poisson arrival model.
    Let $m=(1+\Lambda^{-\frac13})\Lambda$.
    
    Let $A_P'$ be the online algorithm in the Poisson arrival model that acts the same as $A_P$ except that it will discard all online vertices after the $m$-th one. When $n\le m$, $A_P'$ does the same as $A_P$; when $n>m$, the objective of $A_P'$ is at most $\left\lceil\frac{n-m}{\Lambda}\right\rceil\opt_O$ less than $\alg_P$.
    Since $\E[n]=\textrm{Var}[n]=\Lambda$, we have $\Pr[n>m] \le \Lambda^{-\frac13}$ and $\E[(n-m)^+] \le \Lambda^{\frac12}$, where $x^+$ denotes $\max\{x,0\}$.
    So the objective of $A_P'$ is
    \begin{align*}
        \alg_P' & \ge \alg_P - \E\left[\left\lceil\frac{(n-m)^+}{\Lambda}\right\rceil\opt_O\right] \\
        & \ge \alg_P - \left(\frac{\E[(n-m)^+]}{\Lambda} + \Pr[n>m]\right)\opt_O \\
        & \ge \alg_P - \left(\Lambda^{-\frac12}+\Lambda^{-\frac13}\right)\opt_O.
    \end{align*}

    Given $A_P$, we define $A_O$ as follows. We first sample a Poisson process with $n\sim\mathrm{Poisson}(\Lambda)$ online vertices arrived, without drawing their types.
    Then, we uniformly sample $l=\min\{n,\Lambda\}$ different integers $k_1<k_2<\cdots<k_l$ from $[1,\min\{n,m\}]$ in the beginning. Then, we simulate $A_P$ based on the pre-sampled Poisson process, match the $i$-th online vertex in the same way as the $k_i$-th online vertex in $A_P$ for $i\le l$, and discard remaining online vertices (if there are).

    For any given $n$, the expected objective of $A_O$ is $\frac{\min\{n,\Lambda\}}{\min\{n,m\}}$ times the expected objective of $A_P'$. So the overall expected objective of $A_O$ is at least $\frac{\Lambda}{m}\alg_P' = (1-O(\Lambda^{-\frac13}))\alg_P - O(\Lambda^{-\frac13})\opt_O$.
\end{proof}

\subsection{Missing Details in \cref{thm:hardness}}
\label{app:hardness}

The expected objective of \cref{alg:opt-G} can be computed as follows. Let $f(t)$ denote the probability that exactly one offline vertex has been matched at time $t$ and $g(t)$ denote the probability that both offline vertices have been matched at time $t$.
By our matching strategy in \cref{alg:opt-G}, during time $[0,t_0]$, each offline vertex will only be matched if its first-class neighbor arrives. So for $t\in[0,t_0]$, we have
\begin{align*}
    f(t) & = 2\p{1-e^{-(1-\ln2)t}}e^{-(1-\ln2)t}, \\
    g(t) & = \p{1-e^{-(1-\ln2)t}}^2.
\end{align*}

Now consider $t\in[t_0,t_1]$. When one offline vertex has been matched, the other one will only be matched by its first-class neighbor, so
$$\ddt g(t) = f(t)\cdot(1-\ln2).$$
When both offline vertices are unmatched, as long as any online vertex arrives, one of the offline vertices will be matched. Note that the probability that both offline vertices are unmatched at time $t$ is simply $(1-f(t)-g(t))$. So
$$\ddt f(t) = (1-f(t)-g(t))\cdot2 - \ddt g(t).$$
Combining these differential equations with the initial values $f(t_0)$ and $g(t_0)$, we get, for $t\in[t_0,t_1]$,
$$f(t) = 2^{t+1}e^{-t} - \frac{ 2^{2t_0+1}e^{-2t} + 2^{t+t_0+1}e^{-t-t_0}\ln2}{1+\ln2}$$
and
$$g(t)=1 - 2^{t+1}e^{-t} + \frac{2^{2t_0}e^{-2t}(1-\ln2)  + 2^{t+t_0+1}e^{-t-t_0}\ln2}{1+\ln2}.$$

Now consider $t\in[t_1,1]$. When one offline vertex has been matched, the other one will be matched by both neighbors, so
$$\ddt g(t) = f(t)\cdot(1+\ln2).$$
And we still have
$$\ddt f(t) = (1-f(t)-g(t))\cdot2 - \ddt g(t).$$
Combining these differential equations with the initial values $f(t_1)$ and $g(t_1)$, we get, for $t\in[t_1,1]$,

$$f(t) = 2^{2t_1+1-t}e^{-t} -\frac{2^{2t_0+1}e^{-2t}}{1-\ln2}
-\frac{2^{2t_1+t_0+1}e^{-t-t_0}\ln2}{2^t(1+\ln2)}
+\frac{2^{2t_0+t_1+2}e^{-t-t_1}\ln2}{2^t(1+\ln2)(1-\ln2)}$$
and
$$g(t) = 1 - 2^{2t_1+1-t}e^{-t}
+ \frac{2^{2t_0}e^{-2t}(1+\ln2)}{1-\ln2}
+ \frac{2^{2t_1+t_0+1}e^{-t-t_0}\ln2}{2^t(1+\ln2)}
- \frac{2^{2t_0+t_1+2}e^{-t-t_1}\ln2}{2^t(1+\ln2)(1-\ln2)}.$$

Then we know that the expected number of offline vertices we have matched is $f(1)+2g(1)$. To compute the expected objective, it remains to compute how many first-class edges we have matched. Let $p(t)$ denote the expected number of first-class edges we have matched at time $t$. We have
$$\ddt p(t) = (1-f(t)-g(t))\cdot2(1-\ln2)+f(t)\cdot(1-\ln2).$$
Solving it, we get
$$p(1) = 2 + 2^{2t_0}e^{-2}\ln2 - \frac{2^{2t_1}e^{-1}(1-\ln2)}{1+\ln2} +
\frac{P}{(1+\ln2)^2},$$
where
\begin{align*}
    P = & \ 2^{t_0+2t_1}e^{-t_0-1}(1-\ln2)\ln2 + 2^{2t_0+1}e^{-2t_1}(1-\ln2)\ln2 + 2^{t_0 + t_1 + 2}e^{-t_0 - t_1 }\ln^2 2 \\
    & - 2^{2t_0+t_1+1}e^{-t_1-1}\ln2 - 2^{t_1+2}e^{-t_1}(1+\ln2)\ln2 - 2^{2t_0}e^{-2t_0}(1-\ln2)^2\ln2.
\end{align*}

Then, the expected objective is $\alg(k,t_0,t_1)=f(1)+2g(1)+(k-1)p(1)$.

\subsection{Missing Details in \cref{lem:res}}
\label{app:alg-res}

Given \cref{app:hardness}, it's straightforward to compute all we need. In the end, each first-class edge is matched with probability $q_1=p(1)/2$. Each second-class edge is matched with probability $q_2=(f(1)+2g(2)-p(1))/2$. Forcing $t_1=t_0$ and balancing $q_1/(1-\ln2)=q_2/\ln2$, we get $q_1/(1-\ln2)=q_2/\ln2 \approx 0.66217$ when $t_0 \approx 0.14753$.

\subsection{Computational Aspects of \cref{alg:general}}
\label{app:estimate}

The matching strategy of \cref{alg:general} depends on which offline vertices are still unmatched. If we want to compute $\gb_{u,v}(t)$ precisely, one way is to define a function $f_S(t)$ for each set $S\subseteq J$, and solve the differential equations between them as we did in \cref{app:hardness}. However, there will be exponentially many equations, so we do not expect to solve them efficiently.

Here we discuss how to efficiently estimate $\gb_{u,v}(t)$ with an $o(1)$ multiplicative error.
First, notice that $\gb_{u,v}(t)=\Omega(1)$ for $t\ge t_0$, which means we only need to achieve an $o(1)$ additive error.
Since $\gb_{u,v}(t)$ is a continuous function with bounded derivative, we can approximate $t$ by its closest multiple of $1/poly(\Lambda)$. Therefore, we only need to consider polynomially many discrete time points, which means the total number of $\gb_{u,v}(t)$'s we need to estimate is also polynomial. Then we can use the standard Monte-Carlo method to estimate them in polynomial time, with $o(1)$ additive errors.

\end{document}